\documentclass[english]{ccdconf}

\usepackage[english]{babel}
\usepackage{theorem}
\usepackage[framemethod=tikz]{mdframed}
\theoremheaderfont{\itshape\bfseries}
{\theorembodyfont{\itshape}
\newtheorem{assumption}{\textbf{Assumption}}
\newtheorem{lemma}{\textbf{Lemma}}
\newtheorem{definition}{\textbf{Definition}}
\newtheorem{theorem}{\textbf{Theorem}}

\newtheorem{remark}{\textbf{Remark}}

\newtheorem{problem}{\textbf{Problem}}
}
\usepackage{graphicx}
\usepackage{epsfig}
\usepackage{mathptmx}
\usepackage{times}
\usepackage{amsmath}
\usepackage{amsfonts}
\usepackage{amssymb}
\usepackage{mathrsfs}
\usepackage{xcolor}
\usepackage{graphicx}
\usepackage{color,soul}

\newcommand{\T}{^{\mbox{\tiny T}}}
\newcommand{\R}{\mathbb{R}}

\newcommand{\eps}{\varepsilon}
\let\leq\leqslant
\let\geq\geqslant

\newenvironment{proof}[1][Proof]%
{\par\noindent\textit{#1:\ }}%
{\hspace*{\fill} \rule{6pt}{6pt}}
\newenvironment{proof*}[1][Proof]%
{\par\noindent\textit{#1:\ }}{}

\DeclareMathOperator{\diag}{diag}

\usepackage{tikz}
\usetikzlibrary{shapes,calc,arrows,patterns,decorations.pathmorphing
	,decorations.markings}
\usetikzlibrary{arrows.meta}

\newenvironment{system}[1]%
{\setlength{\arraycolsep}{0.5mm}\left\{ \; \begin{array}{#1}}%
    {\end{array} \right.}
\newenvironment{system*}[1]%
{\setlength{\arraycolsep}{0.5mm} \begin{array}{#1}}%
  {\end{array}}

\begin{document}

\title{Regulated State Synchronization for Homogeneous Networks of Non-introspective Agents in Presence of Input Delays: A Scale-Free Protocol Design}
\author{Zhenwei Liu\aref{neu}, Donya Nojavanzadeh\aref{wsu}, Dmitri Saberi\aref{stu}, Ali Saberi\aref{wsu},
	Anton A. Stoorvogel\aref{ut}}

\affiliation[neu]{College of Information Science and
	Engineering, Northeastern University, Shenyang 110819, China
	\email{liuzhenwei@ise.neu.edu.cn}}
\affiliation[wsu]{School of Electrical Engineering and Computer
	Science, Washington State University, Pullman, WA 99164, USA
	\email{donya.nojavanzadeh@wsu.edu; saberi@eecs.wsu.edu}}
\affiliation[stu]{Stanford University, Stanford, CA 94305, USA\email{dsaberi@stanford.edu}}
\affiliation[ut]{Department of Electrical Engineering,
	Mathematics and Computer Science, University of Twente, Enschede, The Netherlands
	\email{A.A.Stoorvogel@utwente.nl}}

\maketitle

\begin{abstract}
   This paper studies regulated state synchronization of homogeneous networks of non-introspective agents in presence of unknown nonuniform input delays. A scale free protocol is designed based on additional information exchange, which does not need any knowledge of the directed network topology and the spectrum of the associated Laplacian matrix. The proposed protocol is scalable and achieves state synchronization for any arbitrary number of agents. Meanwhile, an upper bound for the input delay tolerance is obtained, which explicitly depends on the agent dynamics.
  \end{abstract}

	
	\footnotetext{This work is supported by Nature Science
		Foundation of Liaoning Province under Grant 2019-MS-116.}
\section{Introduction}

The synchronization problem of multi-agent systems (MAS) has attracted
substantial attention during the past decade, due to the wide potential for
applications in several areas such as automotive vehicle control,
satellites/robots formation, sensor networks, and so on. See for
instance the books \cite{ren-book} and \cite{wu-book} or the
survey paper \cite{saber-murray3}.

State synchronization inherently requires homogeneous networks
(i.e. agents which have identical dynamics). Therefore, in this paper we
focus on homogeneous networks. So far, most work has focused on state
synchronization based on diffusive full-state coupling, where the
agent dynamics progress from single- and double-integrator dynamics
(e.g.  \cite{saber-murray2}, \cite{ren}, \cite{ren-beard}) to more
general dynamics (e.g. \cite{scardovi-sepulchre}, \cite{tuna1},
\cite{wieland-kim-allgower}). State synchronization based on
diffusive partial-state coupling has also been considered, including static design (\cite{liu-zhang-saberi-stoorvogel-auto} and \cite{liu-zhang-saberi-stoorvogel-ejc}), dynamic design (\cite{kim-shim-back-seo}, \cite{seo-back-kim-shim-iet}, \cite{seo-shim-back}, \cite{su-huang-tac},
\cite{tuna3}), and the design with additional communication (\cite{chowdhury-khalil} and \cite{scardovi-sepulchre}). 

Meanwhile, time-delay
effects are ubiquitous in any communication scheme. As clarified in \cite{cao-yu-ren-chen}, we can
identify two kinds of delays. Firstly, there is the notion of a \emph{communication
	delay}, which results from limitations on the communication between
agents. Secondly, we have the \emph{input delay}, which is due to
computational limitations of an individual agent. Many works have
focused on dealing with input delay, progressing from single- and
double-integrator agent dynamics (see e.g.\ \cite{saber-murray2},
\cite{tian-liu}, \cite{tian}, \cite{xiao-wang}) to more general agent
dynamics (see e.g. \cite{stoorvogel-saberi-acc},
\cite{nonidentical-delay-c}, \cite{consensus-identical-delay-c},
\cite{zhang-saberi-stoorvogel-cdc15}, \cite{liu-stoorvogel-saberi-zhang-cdc17}). Its objective is to derive an
upper bound on the input delay such that agents can still achieve
synchronization. Moreover, such an upper bound always depends on the
agent dynamics and the network properties.


%
In this paper, we deal with regulated state synchronization problem for MAS in presence of unknown nonuniform input delays by tracking the trajectory of an exosystem. Meanwhile, we can obtain an upper bound for the input delay tolerance, which only depends on the agent dynamics. We design dynamic protocols by using additional information exchange for MAS with both full- and partial-state coupling. The protocol design is scalefree, namely
\begin{itemize}
	\item  The design is independent of information about communication networks. That is to say, the dynamical protocol can work for any communication network such that all of its nodes have path to the exosystem. 
\item  The dynamic protocols are designed for networks with input delays where the admissible upper bound on delays only depends on agent model and does not depend on communication network and the number of agents.
\item The proposed protocols archive regulated state synchronization for any MAS with any number of agents, any admissible non uniform input delays,
and any communication network.
\end{itemize}

\subsection*{Notations and definitions}

Given a matrix $A\in \mathbb{R}^{m\times n}$, $A\T$ denotes its
conjugate transpose and $\|A\|$ is the induced 2-norm. Let $j$ indicate $\sqrt{-1}$. A square matrix
$A$ is said to be Hurwitz stable if all its eigenvalues are in the
open left half complex plane. We denote by
$\diag\{A_1,\ldots, A_N \}$, a block-diagonal matrix with
$A_1,\ldots,A_N$ as the diagonal elements. $A\otimes B$ depicts the
Kronecker product between $A$ and $B$. $I_n$ denotes the
$n$-dimensional identity matrix and $0_n$ denotes $n\times n$ zero
matrix; sometimes we drop the subscript if the dimension is clear from
the context. Moreover let $\mathfrak{C}_\tau^n=\mathit{C}([-\tau,0],\mathbb{R}^n)$ denote the Banach space of all continues functions from $[-\tau,0]\to \mathbb{R}^n$ with norm
\[
\|x\|_{\mathit{C}}=\sup_{t\in[-\tau,0]}\|x(t)\|.
\]

To describe the information flow among the agents we associate a \emph{weighted graph} $\mathcal{G}$ to the communication network. The weighted graph $\mathcal{G}$ is defined by a triple
$(\mathcal{V}, \mathcal{E}, \mathcal{A})$ where
$\mathcal{V}=\{1,\ldots, N\}$ is a node set, $\mathcal{E}$ is a set of
pairs of nodes indicating connections among nodes, and
$\mathcal{A}=[a_{ij}]\in \mathbb{R}^{N\times N}$ is the weighted adjacency matrix with non negative elements $a_{ij}$. Each pair in $\mathcal{E}$ is called an \emph{edge}, where
$a_{ij}>0$ denotes an edge $(j,i)\in \mathcal{E}$ from node $j$ to
node $i$ with weight $a_{ij}$. Moreover, $a_{ij}=0$ if there is no
edge from node $j$ to node $i$. We assume there are no self-loops,
i.e.\ we have $a_{ii}=0$. A \emph{path} from node $i_1$ to $i_k$ is a
sequence of nodes $\{i_1,\ldots, i_k\}$ such that
$(i_j, i_{j+1})\in \mathcal{E}$ for $j=1,\ldots, k-1$. A \emph{directed tree} is a subgraph (subset
of nodes and edges) in which every node has exactly one parent node except for one node, called the \emph{root}, which has no parent node. The \emph{root set} is the set of root nodes. A \emph{directed spanning tree} is a subgraph which is
a directed tree containing all the nodes of the original graph. If a directed spanning tree exists, the root has a directed path to every other node in the tree.  

For a weighted graph $\mathcal{G}$, the matrix
$L=[\ell_{ij}]$ with
\[
\ell_{ij}=
\begin{system}{cl}
\sum_{k=1}^{N} a_{ik}, & i=j,\\
-a_{ij}, & i\neq j,
\end{system}
\]
is called the \emph{Laplacian matrix} associated with the graph
$\mathcal{G}$. The Laplacian matrix $L$ has all its eigenvalues in the
closed right half plane and at least one eigenvalue at zero associated
with right eigenvector $\textbf{1}$ \cite{royle-godsil}. Moreover, if the graph contains a directed spanning tree, the Laplacian matrix $L$ has a single eigenvalue at the origin and all other eigenvalues are located in the open right-half complex plane \cite{ren-book}.

\section{Problem formulation}

Consider a MAS consisting of $N$ identical linear dynamic agents with input delays:
\begin{equation}\label{eq1}
\begin{cases}
\dot{x}_i(t)=Ax_i(t)+Bu_i(t-\tau_i),\\
y_i(t)=Cx_i(t),\\
x_i(\delta)=\phi_i(\delta), \quad \delta\in[-\bar{\tau},0]
\end{cases}
\end{equation}
where $x_i(t)\in\mathbb{R}^n$, $y_i(t)\in\mathbb{R}^q$ and
$u_i(t)\in\mathbb{R}^m$ are the state, output, and the input of agent 
$i=1,\ldots, N$, respectively. $\tau_i$ represent the input delays with $\tau_i\in[0,\bar{\tau}]$, where $\bar{\tau} = \max_{i} \tau_i$ and $\phi_i \in \mathfrak{C}_{\bar{\tau}}^n$.

\begin{assumption}\label{Aass}
	We assume that:
	\begin{itemize}		
		\item[(i)] $(A,B)$ are stabilizable and $(C,A)$ are detectable.
		\item[(ii)] All eigenvalues of $A$ are in the closed left half plane. 
	\end{itemize}
\end{assumption}

The network provides agent $i$ with the following information,
\begin{equation}\label{eq2}
\zeta_i(t)=\sum_{j=1}^{N}a_{ij}(y_i(t)-y_j(t)),
\end{equation}
where $a_{ij}\geq 0$ and $a_{ii}=0$. This communication topology of
the network can be described by a weighted graph $\mathcal{G}$ associated with \eqref{eq2}, with
the $a_{ij}$ being the coefficients of the weighting matrix
$\mathcal{A}$ not of the dynamics matrix $A$ introduced in\eqref{eq1}). In terms of the coefficients of the associated
Laplacian matrix $L$, $\zeta_i$ can be rewritten as
\begin{equation}\label{zeta_l}
\zeta_i(t) = \sum_{j=1}^{N}\ell_{ij}y_j(t).
\end{equation}
We refer to this as \emph{partial-state coupling} since only part of
the states are communicated over the network. When $C=I$, it means all states are communicated over the network, we call it \emph{full-state coupling}. Then, the original agents are expressed as
\begin{equation}\label{neq1}
\dot{x}_i(t)=Ax_i(t)+Bu_i(t-\tau_i)
\end{equation}
and $\zeta_i$ is rewritten as
\[
\zeta_i(t) = \sum_{j=1}^{N}\ell_{ij}x_j(t).
\]

Obviously, state synchronization is achieved if
\begin{equation}\label{synch_org}
\lim_{t\to \infty} (x_i(t)-x_j(t))=0.
\end{equation}
for all $i,j \in {1,...,N}$.


In this paper, we consider regulated
state synchronization. The reference trajectory is generated by the following exosystem
\begin{equation}\label{solu-cond}
\begin{system*}{rl}
\dot{x}_r(t) & = A x_r(t)\\
y_r(t)&=Cx_r(t).
\end{system*}	
\end{equation}
with $x_r\in\R^n$.  Our objective is that the agents achieve regulated
state synchronization, that is
\begin{equation}\label{synchro}
\lim_{t\to \infty} (x_i(t)-x_r(t))=0,
\end{equation}
for all $i\in\{1,\ldots,N\}$. Clearly, we need some level of
communication between the exosystem and the agents.  We assume that a nonempty
subset $\mathcal{C}$ of the agents have access to their
own output relative to the output of the exosystem.  Specially, each
agent $i$ has access to the quantity
\begin{equation}\label{elp}
\psi_i=\iota_{i}(y_i(t)-y_r(t)),\qquad
\iota_i=
\begin{cases}
1, & i\in \mathcal{C},\\
0, & i\notin \mathcal{C}.
\end{cases}
\end{equation}
By combining this with \eqref{zeta_l}, we have the following information
exchange
\begin{equation}\label{zeta_l-n}
\bar{\zeta}_i(t) = \sum_{j=1}^{N}a_{ij}(y_i(t)-y_j(t))+\iota_{i}(y_i(t)-y_r(t)).
\end{equation}
Meanwhile, \eqref{zeta_l-n} will change as
\begin{equation}\label{zeta_l-nn}
\bar{\zeta}_i (t)= \sum_{j=1}^{N}a_{ij}(x_i(t)-x_j(t))+\iota_{i}(x_i(t)-x_r(t))
\end{equation}
for full-state coupling case.

To guarantee that each agent can achieve the required regulation, we need to make sure that there exists a path to each node starting with node from the set $\mathscr{C}$. Motivated by this requirement, we define the following set of graphs.

\begin{definition}\label{graph-def}
	Given a node set $\mathscr{C}$, we denote by $\mathbb{G_\mathscr{C}^N}$ the set of all graphs with $N$ nodes containing the node set $\mathscr{C}$, such that every node of the network graph $\mathscr{G}\in \mathbb{G_\mathscr{C}^N}$ is a member of a directed tree which has its root contained in the node set $\mathscr{C}$.
\end{definition}
\begin{remark}
	Note that Definition \ref{graph-def} does not require necessarily the existence of directed spanning tree.
\end{remark}
From now on, we will refer to the node set $\mathscr{C}$ as the \emph{root set} in view of Definition \ref{graph-def}. For any graph $\mathcal{G} \in \mathbb{G_\mathscr{C}^N}$, with Laplacian matrix $L$, we define the expanded Laplacian matrix as:
\[
\bar{L}=L+\diag\{\iota_i\}=[\bar{l}_{ij}]_{N\times N}
\]
which is not a regular Laplacian matrix associated to the graph, since the sum of its rows need not be zero. We observe that Definition \ref{graph-def} guarantees that all the eigenvalues of $\bar{L}$ have positive real parts. In particular, we have that $\bar{L}$ is invertible.

In this paper, we introduce an additional
information exchange among protocols. In particular, each agent 
$i=1,\ldots, N$ has access to additional information, denoted by
$\hat{\zeta}_i$, of the form
\begin{equation}\label{eqa1}
\hat{\zeta}_i(t)=\sum_{j=1}^Na_{ij}(\xi_i(t)-\xi_j(t))
\end{equation}
where $\xi_j\in\mathbb{R}^n$ is a variable produced internally by agent $j$ and to be defined in next sections.

We formulate the problem for regulated state
synchronization of a MAS with full- and partial-state coupling.

\begin{problem}\label{prob3}
	Consider a MAS described by \eqref{eq1} satisfying Assumption \ref{Aass}, with a given $\bar{\tau}$ and the
	associated exosystem \eqref{solu-cond}. 
	Let a set of nodes
	$\mathcal{C}$ be given which defines the set
	$\mathbb{G}_{\mathcal{C}}^N$ and 	
	let the asssociated network
	communication graph $\mathcal{G} \in \mathbb{G}_{\mathcal{C}}^N$ be
	given by \eqref{zeta_l-n}.
	
	The \textbf{scalable regulated state synchronization problem with additional information exchange} of a MAS is to find,
	if possible, a linear dynamic protocol for each agent $i \in \{1,\hdots,N\}$, using only knowledge of agent model, i.e., $(A,B,C)$, and upper bound of delays $\bar{\tau}$, of the form:
	\begin{equation}\label{protoco4}
	\begin{system}{cl}
	\dot{x}_{c,i}(t)&=A_{c,i}x_{c,i}(t)+B_{c,i}u_i(t-\tau_i)+C_{c,i}\bar{\zeta}_i(t)+D_{c,i}\hat{\zeta}_i(t),\\
	u_i(t)&=F_{c,i}x_{c,i}(t),
	\end{system}
	\end{equation}
	where $\hat{\zeta}_i(t)$ is defined in \eqref{eqa1} with $\xi_i(t)=H_{c}x_{i,c}(t)$, and $x_{c,i}(t)\in\R^{n_i}$, such that regulated
	state synchronization \eqref{synchro} is achieved for any $N$ and any
	graph $\mathcal{G}\in \mathbb{G}_{\mathcal{C}}^N$.
\end{problem}

%

\section{Protocol Design}

In this section, we will consider the regulated state synchronization problem for a MAS with input delays. In particular, we cover separately systems with full-state coupling and those with partial-state coupling.

\subsection{Full-state coupling}

Firstly, we define
\[
\omega_{\max}=
\begin{cases}
0, & \text{A is Hurwitz},\\
\max\{\omega\in\R | \det(j\omega I-A)=0\}, &\text{otherwise}.
\end{cases}
\]

Then, we design a dynamic protocol with additional information exchanges for agent
$i\in\{1,\ldots,N\}$ as follows.
\begin{equation}\label{pscp1}
\begin{system}{cll}
\dot{\chi}_i(t) &=& A\chi_i(t)+Bu_i(t-\tau_i)+\bar{\zeta}_i(t)-\hat{\zeta}_i(t)-\iota_{i}\chi_i(t) \\
u_i(t) &=& - \rho B\T P_{\eps}\chi_i(t),
\end{system}
\end{equation}
where $\rho>0$ and
$\eps$ is a parameter satisfying
$\eps\in (0,1]$, $P_{\eps}$ satisfies
\begin{equation}\label{arespecial}
A\T P_{\eps} + P_{\eps} A -  P_{\eps} BB\T
P_{\eps} + \eps I = 0 
\end{equation}
Note that for any $\eps>0$, there exists a unique solution of \eqref{arespecial}.

The agents communicate $\xi_i$, which are chosen as $\xi_i(t)=\chi_i(t)$, therefore each agent has access to the following information:
\begin{equation}\label{info1}
\hat{\zeta}_i(t)=\sum_{j=1}^Na_{ij}(\chi_i(t)-\chi_j(t)).
\end{equation}
while $\bar{\zeta}_i(t)$ is defined by \eqref{zeta_l-nn}.\\

Our formal result is stated in the following theorem.
\begin{theorem}\label{mainthm1}
	Consider a MAS described by \eqref{neq1} satisfying Assumption \ref{Aass}, with a given $\bar{\tau}$ and the
	associated exosystem \eqref{solu-cond}. 
	Let a set of nodes
	$\mathcal{C}$ be given which defines the set
	$\mathbb{G}_{\mathcal{C}}^N$ and 	
	let the asssociated network
	communication graph $\mathcal{G} \in \mathbb{G}_{\mathcal{C}}^N$ be
	given by \eqref{zeta_l-nn}.
	
	Then the scalable regulated state synchronization problem as stated in Problem
	\ref{prob3} is solvable if
	\begin{equation}\label{boundtau}
	\bar{\tau}\omega_{\max}<\frac{\pi}{2}. 
	\end{equation}
	In particular, there exist $\rho>1$ and $\eps^* > 0$ that depend only on $\bar{\tau}$ and the agent models such that, for any $\eps\in(0,\eps^*]$, the dynamic protocol given by \eqref{pscp1} and \eqref{arespecial} solves the scalable regulated state
	synchronization problem for any $N$ and any graph
	$\mathcal{G}\in\mathbb{G}_{\mathcal{C}}^N$.
\end{theorem}

To obtain this result, we need the following lemma.

\begin{lemma}[\cite{consensus-identical-delay-c}]\label{hode-lemma-ddelay}
	Consider a linear time-delay system
	\begin{equation}\label{hode-lemma-system-c}
	\dot{x}(t)=Ax(t)+\sum_{i=1}^{m}A_{i}x(t-\tau_{i}),
	\end{equation}
	where $x(t)\in\R^{n}$ and $\tau_{i}\in\mathbb{R}$. Assume that 
	$A+\sum_{i=1}^{m}A_{i}$ is Hurwitz stable. Then,
	\eqref{hode-lemma-system-c} is asymptotically stable for
	$\tau_1,\ldots,\tau_N\in[0,\bar{\tau}]$ if
	\[
	\det[j\omega I-A- \sum_{i=1}^{m}e^{-j\omega\tau_i}A_{i}]\neq 0,
	\]
	for all $\omega\in\R$, and for all
	$\tau_1,\ldots,\tau_N\in[0,\bar{\tau}]$.
\end{lemma}

\begin{proof}[Proof of Theorem \ref{mainthm1}]
	Firstly, let $\tilde{x}_i(t)=x_i(t)-x_r(t)$, we have
	\[
	\dot{\tilde{x}}_i(t)=A{\tilde{x}}_i(t)+Bu_i(t-\tau_i)
	\]

	We define 
	\begin{align*}
	&\tilde{x}(t)=\begin{pmatrix}
	\tilde{x}_1(t)\\\vdots\\\tilde{x}_N(t)
	\end{pmatrix} , 
	\chi(t)=\begin{pmatrix}
	\chi_1(t)\\\vdots\\\chi_N(t)
	\end{pmatrix},  \\
	&\tilde{x}^\tau(t)=\begin{pmatrix}
	\tilde{x}_1(t-\tau_1)\\\vdots\\\tilde{x}_N(t-\tau_N)
	\end{pmatrix},\text{ and }  
	\chi^\tau(t)=\begin{pmatrix}
	\chi_1(t-\tau_1)\\\vdots\\\chi_N(t-\tau_N)
	\end{pmatrix}
	\end{align*}
	then we have the following closed-loop system
	\begin{equation}
	\begin{system*}{ll}
	\dot{\tilde{x}}(t)=&(I\otimes A) \tilde{x}(t)- \rho(I\otimes BB\T P_\eps)\chi^\tau(t)\\
	\dot{\chi}(t)=&(I\otimes A) \chi(t)-\rho (I\otimes BB\T P_\eps)\chi^\tau(t)\\
	&\qquad\qquad+(\bar{L}\otimes I)(\tilde{x}(t)-\chi(t)).
	\end{system*}
	\end{equation}

	Let $e(t)=\tilde{x}(t)-\chi(t)$, we can obtain  
	\begin{equation}\label{newsystem2}
	\begin{system*}{l}
	\dot{\tilde{x}}(t)=(I\otimes A) \tilde{x}(t)- \rho (I\otimes BB\T P_\eps)\tilde{x}^\tau(t)+ \rho (I\otimes BB\T P_\eps)e^\tau(t)\\
	\dot{e}(t)=(I\otimes A-\bar{L}\otimes I)e(t)
	\end{system*}
	\end{equation}
	where $e^\tau(t)=\tilde{x}^\tau(t)-\chi^\tau(t)$.
	
	The proof proceeds in two steps.
	
	\textbf{Step 1:}  First, we prove the stability of system \eqref{newsystem2} without delays, i.e.
	
	\begin{equation}\label{newsystem22}
	\begin{system*}{l}
	\dot{\tilde{x}}=(I\otimes A) \tilde{x}- \rho (I\otimes BB\T P_\eps)\tilde{x}+\rho (I\otimes BB\T P_\eps)e\\
	\dot{e}=(I\otimes A-\bar{L}\otimes I)e
	\end{system*}
	\end{equation}
	
	Since all eigenvalues of $\bar{L}$ are positive, we have
	\begin{equation}\label{boundapl}
	(T\otimes I)(I\otimes A-\bar{L}\otimes I)(T^{-1}\otimes I)=I\otimes A-\bar{J}\otimes I
	\end{equation}
	for a non-singular transformation matrix $T$, where	\eqref{boundapl}  is upper triangular Jordan form with $A-\lambda_i I$ for $i=1,\cdots,N-1$ on the diagonal. Since all eigenvalues of $A$ are in the closed left half plane, $A-\lambda_i I$ is stable. Therefore, all eigenvalues of $I\otimes A-\bar{L}\otimes I$ have negative real part.
	Therefore, we have that the dynamics for $e$ is asymptotically stable.

	According to the above result, for \eqref{newsystem22} we just need to prove the stability of 
	\[
	\dot{\tilde{x}}=[I\otimes (A-\rho BB\T P_\eps)] \tilde{x}
	\]
	or the stability of 
	\[
	A-\rho BB\T P_\eps
	\]
	
	Based on the ARE \eqref{arespecial}, for a positive definite matrix $P_\eps$, we have
	\begin{align*}
	&P_\eps(A- \rho BB\T P_\eps)+(A- \rho BB\T P_\eps)\T P_\eps\\
	\leq &-\eps I-(2\rho -1)P_\eps BB\T P_\eps\\
	< &0
	\end{align*}
	for $\eps >0$ and $\rho >1$.
	
	\textbf{Step2:} In this step, since we have that $e_i$ is asymptotically stable, we just need to prove the stability of
	\[
	\dot{\tilde{x}}_i(t)= A \tilde{x}_i(t)- \rho BB\T P_\eps\tilde{x}_i(t-\tau_i)
	\]
	for $i=1,\hdots, N$. Following Lemma \ref{hode-lemma-ddelay}
	we need to prove
	\begin{equation}\label{bound-condition}
	\det[j\omega I-A+ \rho e^{-j\omega\tau_i}BB\T P_\eps]\neq 0
	\end{equation}
	for $\omega\in \R$ and $\tau_i\in[0,\bar{\tau}]$. We choose $\rho$, such that  
	\begin{equation}	
	\rho>\frac{1}{\cos(\bar{\tau}\omega_{\max})}.
	\end{equation}
	Let $\rho$ be fixed. Meanwhile, we note that there exists a $\theta$ such that
	\[
	\rho>\frac{1}{\cos(\bar{\tau}\omega)}, \forall |\omega|<\omega_{\max}+\theta
	\]
	
	Then, we split the proof of \eqref{bound-condition} into
	two cases where $|\omega|\geq \omega_{\max}+\theta$ and $|\omega|<\omega_{\max}+\theta$ respectively.
	
	If $|\omega|\geq \omega_{\max}+\theta$, we have
	$\det(j\omega I-A)\neq 0$, which yields
	$\sigma_{\min}(j\omega I-A)>0$. Hence, we can state that there exists a $\mu>0$ such
	that
	\[
	\sigma_{\min}(j\omega I-A)>\mu,\quad \forall \omega \text{ such
		that } |\omega|\geq \omega_{\max}+\theta.
	\]
	
	The above bound always exists by observing that
	for $|\omega|>\bar{\omega}:=\max\{\|A\|+1, \omega_{\max}+\theta\}$, we have
	$\sigma_{\min}(j\omega I-A)>|\omega|-\|A\|>1$. However, for 
	$\omega_{\max}+\theta<|\omega|<\bar{\omega}$, there exist a $\mu\in(0,1]$, such
	that $\sigma_{\min}(j\omega I-A)\geq\mu$ since $\sigma_{\min}(j\omega I-A)$ depends
	continuously on $\omega$.
	
	Given $\rho$, there exists $\eps^*>0$ such that $\|\rho BB\T P_{\eps}\|\leq \mu/2$. Then, we obtain
	\[
	\sigma_{\min}(j\omega I-A+\rho e^{-j\omega\tau_i}BB\T P_\eps)\geq \mu-\tfrac{\mu}{2}\geq \tfrac{\mu}{2}.
	\]
	Therefore, condition
	\eqref{bound-condition} holds for
	$|\omega|\geq \omega_{\max}+\theta$.
	
	Now, it remains to show that condition \eqref{bound-condition} holds for $|\omega|<\omega_{\max}+\theta$. We find that
	\[
	-\omega\tau_i<|\omega| \bar{\tau}\leq\dfrac{\pi}{2},
	\]
	and hence 
	\[
	\rho \cos(-\omega\tau_i)>\rho \cos(|\omega| \bar{\tau})>1.
	\]
	It implies that we have
	\[
	A- \rho e^{-j\omega\tau_i}BB\T P_\eps
	\]
	is Hurwitz stable for a positive definite matrix $P_\eps$ (see \cite[Lemma C.1]{consensus-identical-delay-c}). Therefore,  \eqref{bound-condition} 
	holds for $|\omega|<\omega_{\max}+\theta$. Thus, we can obtain the regulated state synchronization result based on Lemma \ref{hode-lemma-ddelay} by choosing
	\[
	\rho>\frac{1}{\cos(\bar{\tau}\omega_{\max} )}.
	\] 	
\end{proof}

\subsection{Partial-state coupling}

In this subsection, we will consider the case via partial-state coupling.
We design the following dynamic protocol with additional information exchanges as follows.
\begin{equation}\label{pscp3}
\begin{system}{cll}
\dot{\hat{x}}_i(t) &=& A\hat{x}_i(t)+B\Phi_i^\tau+K(\bar{\zeta}_i(t)-C\hat{x}_i(t))+\iota_i Bu_i(t-\tau_i) \\
\dot{\chi}_i(t) &=& A\chi_i(t)+Bu_i(t-\tau_i)+\hat{x}_i(t)-\check{\zeta}_i(t)-\iota_{i}\chi_i(t) \\
u_i(t) &=& -\rho B\T P_{\eps}\chi_i(t),
\end{system}
\end{equation}
for $i=1,\ldots,N$ where $K$ is a pre-design matrix such that $A-KC$ is Hurwitz stable and $\rho>0$. $\eps$ is a parameter satisfying
$\eps\in (0,1]$, $P_{\eps}$ satisfies \eqref{arespecial} and is the unique solution of \eqref{arespecial} for any $\eps>0$. $\rho$ and $\omega_{\max}$ are defined in \eqref{pscp1}.
In this protocol, the agents communicate $\chi_i$ and $u_i(t-\tau_i)$, i.e. each agent has access to the following additional information:
\begin{equation}\label{add_1}
\check{\zeta}_i(t)=\sum_{j=1}^Na_{ij}(\chi_i(t)-\chi_j(t)),
\end{equation}
and
\begin{equation}\label{add_2}
\Phi_i^\tau(t)=\sum_{j=1}^{N}a_{ij}(u_i(t-\tau_i)-u_j(t-\tau_j)).
\end{equation}
$\bar{\zeta}_i(t)$ is also defined as \eqref{zeta_l-n}. Then we have the following theorem for MAS via partial-state coupling.

\begin{theorem}\label{mainthm2}
	Consider a MAS described by \eqref{eq1} satisfying Assumption \ref{Aass}, with a given $\bar{\tau}$ and the
	associated exosystem \eqref{solu-cond}. 
	Let a set of nodes
	$\mathcal{C}$ be given which defines the set
	$\mathbb{G}_{\mathcal{C}}^N$ and 	
	let the asssociated network
	communication graph $\mathcal{G} \in \mathbb{G}_{\mathcal{C}}^N$ be
	given by \eqref{zeta_l-n}.
	
	Then the scalable regulated state synchronization problem as stated in Problem
	\ref{prob3} is solvable if \eqref{boundtau} holds.
	In particular, there exist $\rho>1$ and $\eps^* > 0$ that depend only on $\bar{\tau}$ and the agent models such that, for any $\eps\in(0,\eps^*]$, the dynamic protocol given by \eqref{pscp3} and \eqref{arespecial} solves the scalable regulated state
	synchronization problem for any $N$ and any graph
	$\mathcal{G}\in\mathbb{G}_{\mathcal{C}}^N$.
\end{theorem}

\begin{proof}[Proof of Theorem \ref{mainthm2}]
	Similar to Theorem \ref{mainthm1}, let $\tilde{x}_i(t)=x_i(t)-x_r(t)$, we have
	\[
	\begin{system}{cll}
	\dot{\tilde{x}}_i(t)&=&A{\tilde{x}}_i(t)+Bu_i(t-\tau_i)\\
	\dot{\hat{x}}_i(t) &=& A\hat{x}_i(t)+B\Phi_i^\tau(t)+K(\bar{\zeta}_i(t)-C\hat{x}_i(t))+\iota_i Bu_i(t-\tau_i) \\
	\dot{\chi}_i(t) &=& A\chi_i(t)+Bu_i(t-\tau_i)+\hat{x}_i(t)-\check{\zeta}_i(t)-\iota_{i}\chi_i (t)
	\end{system}
	\]
	
	Then we have the following closed-loop system
	\begin{equation}
	\begin{system*}{l}
	\dot{\tilde{x}}(t)=(I\otimes A) \tilde{x}(t)-\rho (I\otimes BB\T P_\eps)\chi^\tau(t)\\
	\dot{\hat{x}}(t) = I\otimes (A-KC)\hat{x}(t)-\rho(\bar{L}\otimes B B\T P_\eps)\chi^\tau(t)+(\bar{L}\otimes KC)\tilde{x}(t) \\
	\dot{\chi} (t)= (I\otimes A-\bar{L}\otimes I)\chi(t)-\rho(I\otimes BB\T P_\eps)\chi^\tau(t)+\hat{x}(t)
	\end{system*}
	\end{equation}
	
	By defining $e(t)=\tilde{x}(t)-\chi(t)$ and $\bar{e}(t)=(\bar{L}\otimes I)\tilde{x}(t)-\hat{x}(t)$, we can obtain  
	\begin{equation}\label{newsystem3}
	\begin{system*}{l}
	\dot{\tilde{x}}(t)=(I\otimes A) \tilde{x}(t)-\rho(I\otimes BB\T P_\eps)\tilde{x}^\tau(t)+\rho(I\otimes BB\T P_\eps)e^\tau(t)\\
	\dot{\bar{e}}(t)=I\otimes (A-KC)\bar{e}(t)\\
	\dot{e}(t)=(I\otimes A-\bar{L}\otimes I)e(t)+\bar{e}(t)
	\end{system*}
	\end{equation}
	
	Similar to Theorem \ref{mainthm1}, we prove the stability of \eqref{newsystem3} without delays first,
	\begin{equation}\label{newsystem33}
	\begin{system*}{l}
	\dot{\tilde{x}}(t)=(I\otimes A) \tilde{x}(t)-\rho(I\otimes BB\T P_\eps)\tilde{x}(t)+ \rho(I\otimes BB\T P_\eps)e(t)\\
	\dot{\bar{e}}(t)=I\otimes (A-KC)\bar{e}(t)\\
	\dot{e}(t)=(I\otimes A-\bar{L}\otimes I)e(t)+\bar{e}(t)
	\end{system*}
	\end{equation}
	
	Since we have $A-KC$ and $I\otimes A-\bar{L}\otimes I$ are Hurwitz stable, one can obtain
	\[
	\lim_{t\to \infty}\bar{e}(t)\to 0 \text{ and }\lim_{t\to \infty}e(t)\to 0
	\]
	i.e. we just need to prove the stability of 
	\[		
	\dot{\tilde{x}}_i(t)=(A-\rho BB\T P_\eps)\tilde{x}_i(t).
	\]		
	Similarly to Theorem \ref{mainthm1}, we can obtain the result about the stability of the above system directly.
	
	Then, since $e_i(t)$ and $\bar{e}_i(t)$ are all asymptotically stable, we just need to prove the stability 
	\[
	\dot{\tilde{x}}_i(t)= A \tilde{x}_i(t)-\rho BB\T P_\eps\tilde{x}_i(t-\tau_i)
	\]
	or prove \eqref{bound-condition} for $\omega\in \R$ and $\tau_i\in[0,\bar{\tau}]$ based on Lemma \ref{hode-lemma-ddelay}. 
	
	Similar to the proof of Theorem \ref{mainthm1}, we can obtain the synchronization result.
\end{proof}

\section{MATLAB Implementation and Example}

\subsection{Implementation}
In this subsection, we present and discuss our implementation of the problem at hand, allowing us to visualize and put to use the theory that we have developed. We discuss two main files that we include in our source code folder, along with one helper file, all implemented with MATLAB. Please see {\tt https://github.com/wsucontrolgroup/ scale-free-input-delay.git}.

The file \verb+protocol_design.m+ designs the central product of this paper, the protocol, setting it up for use. It accepts \emph{only} the agent model ($A,B,C$) and an upper bound on the delays ($\bar{\tau}$). Recall that in the full-state coupling case, the protocol is given by \eqref{pscp1} and \eqref{arespecial} and in the partial-state coupling case by \eqref{pscp3} and \eqref{arespecial}.
In view of these equations, the function \verb+protocol_design+ returns the relevant data necessary to define them. Namely, $\epsilon^*$, $\rho,$ $P_{\epsilon},$ and $K$. Note that the selection of $K$ is arbitrary, and we welcome the user to change our code to pick any value as far as $A - KC$ is Hurwitz stable. 

We describe briefly how this function operates. The proof of Theorem 1, particularly the definitions of $\epsilon^*$ and $\rho$, reveal that there is a large degree of freedom in choosing the pair $(\epsilon^*, \rho).$ Furthermore, different choices can lead to drastically different speeds of convergence (i.e., how quickly $x_i$ converges to $x_r$ for $i = 1, \dots, N$). For fixed $\rho$, among valid choices of $\epsilon^*$, faster convergence is typically obtained by larger $\epsilon^*.$ In this vein, our algorithm seeks to obtain a less conservative $\epsilon^*$, given a fixed $\rho$. This is done by first choosing $\theta$ based on $\rho$, from which we make a non-conservative estimate of $\mu$, finally choosing $\epsilon^*$ based on its definition (which involves $\mu$). Moreover, we comment that this file in no way chooses the best $(\epsilon^*, \rho)$ pair that optimizes convergence. We simply guarantee that our parameters satisfy the solvability conditions laid out in this paper. The existence of an algorithm that chooses optimal $(\epsilon^*, \rho)$ in all generality remains an open question, and will be the subject of future research.

The second and final main file we include is the \verb+ input_delay_solver.m+  file, which is a complete simulation package. This file defines a function of the same name that accepts an agent model, the delays, the adjacency matrix of the communication network, the set of leader nodes, and the initial conditions (in addition to the period of integration $T_{\max}$, which specifies the time interval $(0,T_{\max})$ that the user wants the solution over). From here, the function uses \verb+protocol_design+ to choose acceptable protocol parameters to achieve regulated state synchronization as stated in \eqref{synch_org} through the use of protocols \eqref{pscp1} and \eqref{arespecial} in the case of full-state coupling and \eqref{pscp3} and \eqref{arespecial} for partial-state coupling. If matrix $C$ passed to the function is the identity, the protocol for full-state coupling will be enacted, otherwise, partial state coupling protocol will be utilized.  We make implicit use of the MATLAB \verb+dde23+  solver. Finally, a time series, along with the values of the states, exosystem, and the input at each time step, are all extracted from the resulting structure and returned. This gives data which can be easily plotted or used for a variety of purposes.

We also include a file \verb+initial_conditions.m+ that allows the user to format their initial conditions in a particular way that will be accepted by our solver. We hope the inclusion of these files will allow the reader to illustrate our results for themselves if so desired, and more importantly, that those who have use for a product of this form will profit from them in their own endeavors.

\begin{figure}[t]
	\includegraphics[width=4cm, height=2.5cm]{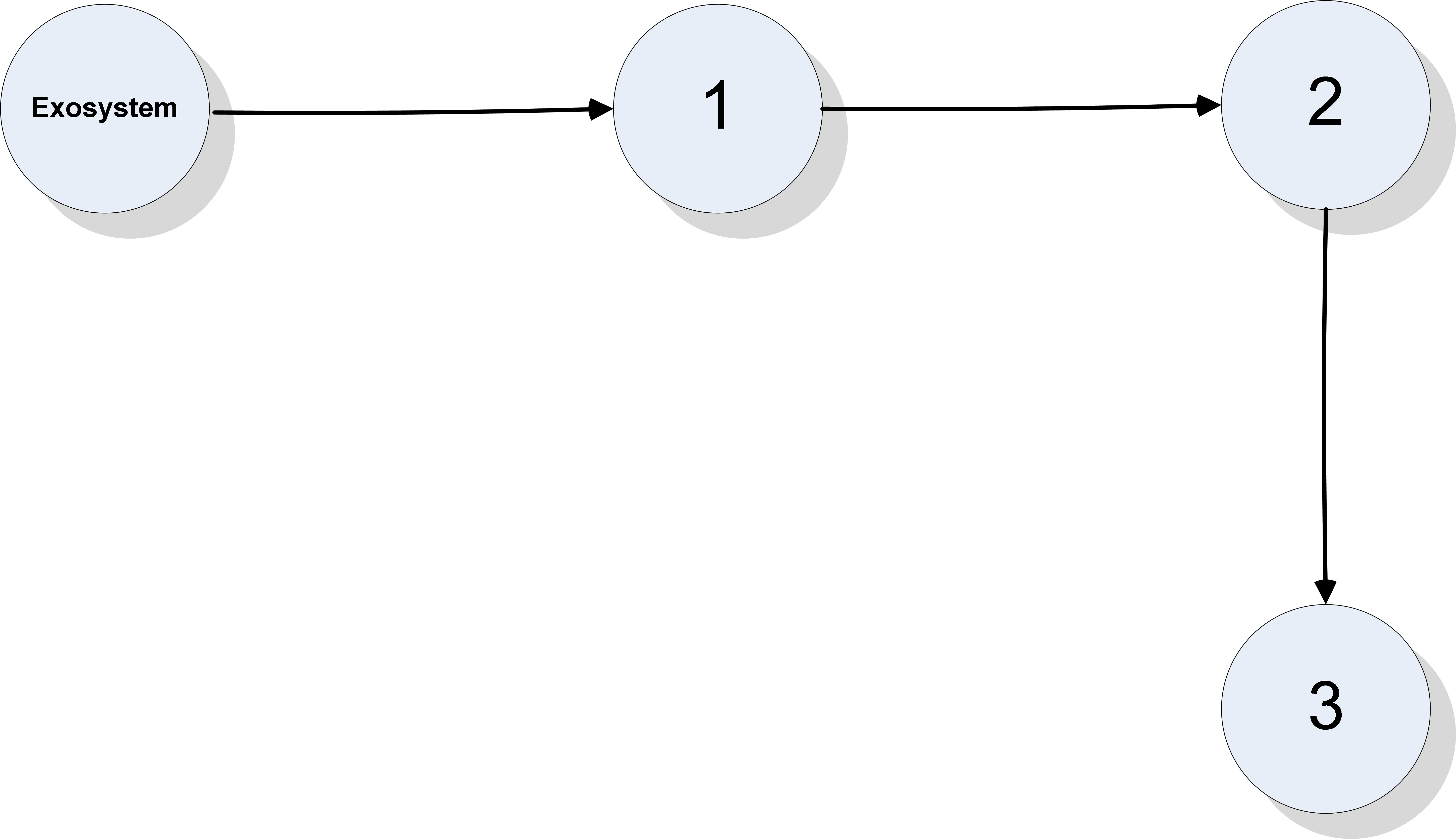}
	\centering
	\caption{Directed communication network with $3$ nodes}\label{graph_3nodes}
\end{figure}

\begin{figure}[t]
	\includegraphics[width=5cm, height=3cm]{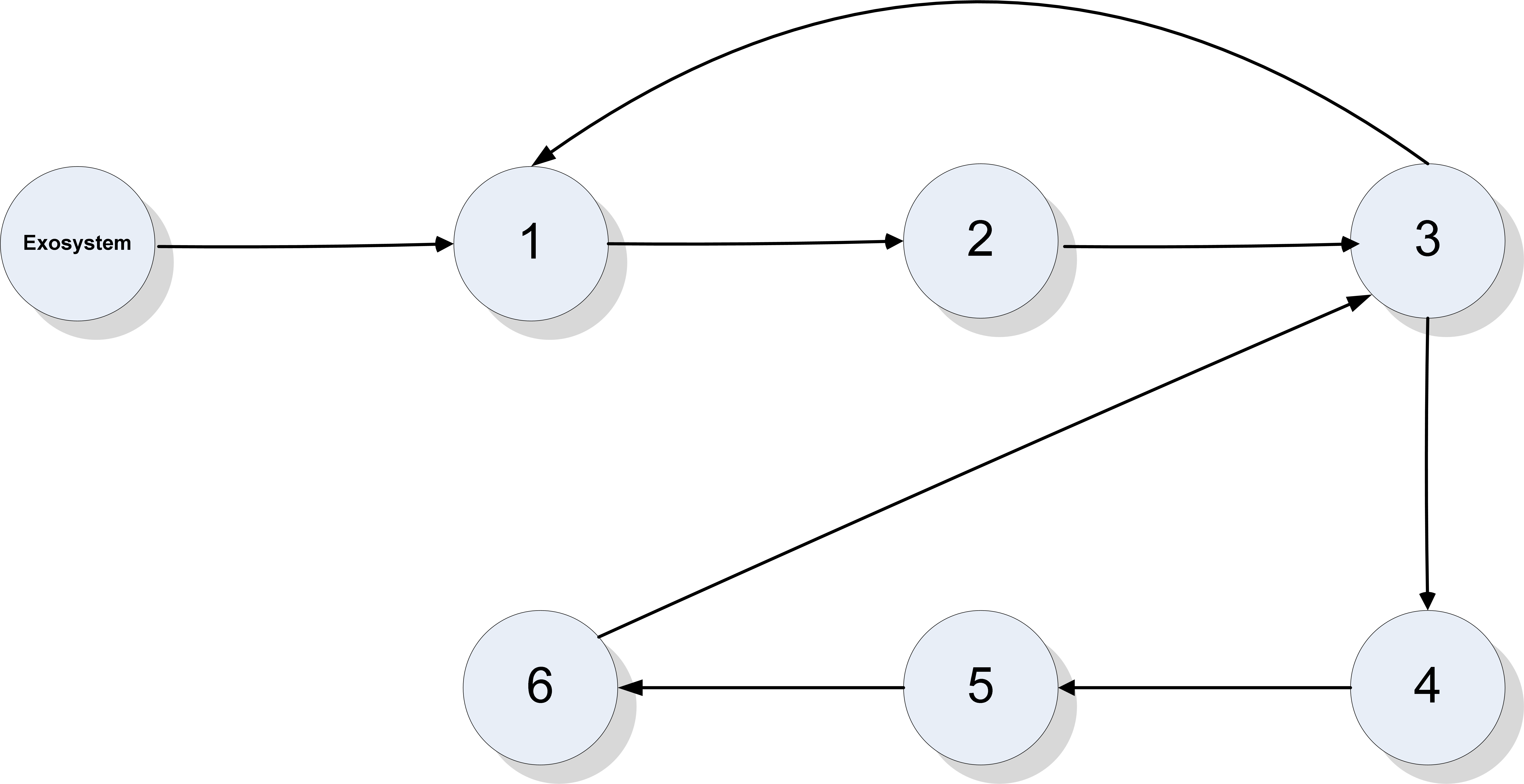}
	\centering
	\caption{Directed communication network with $6$ nodes}\label{graph_6nodes}
\end{figure}

\begin{figure}[t!]
	\includegraphics[width=8cm, height=3cm]{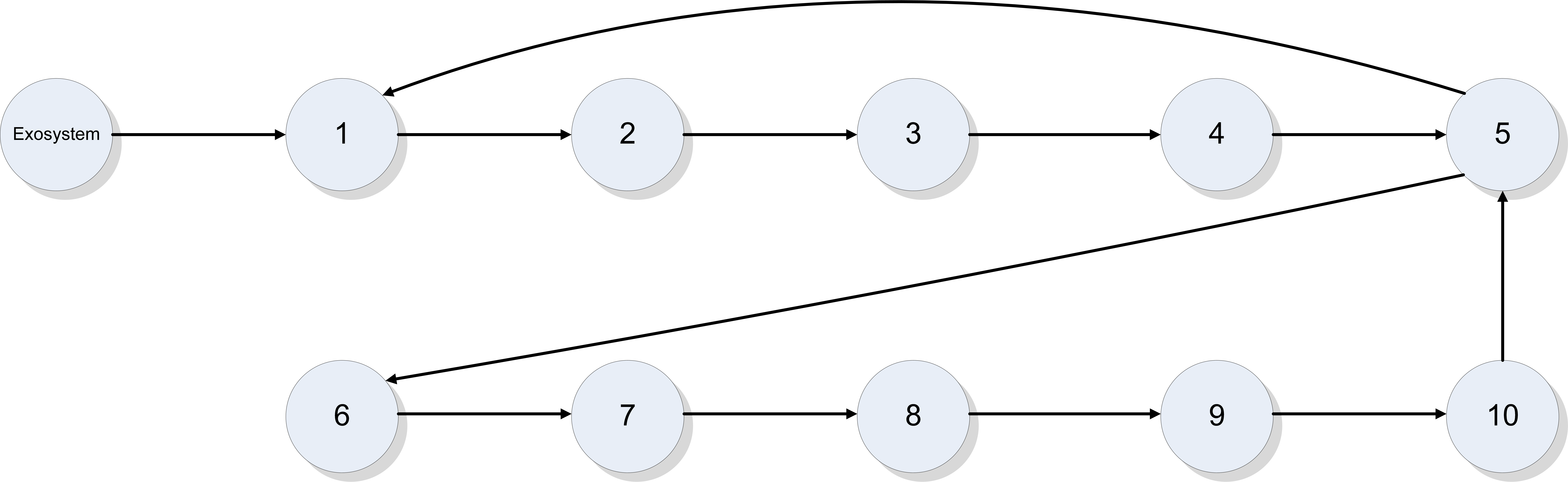}
	\centering
	\caption{Directed communication network with $10$ nodes}\label{graph_10nodes}
\end{figure}

\begin{figure}[t!]
	\includegraphics[width=5cm, height=2.5cm]{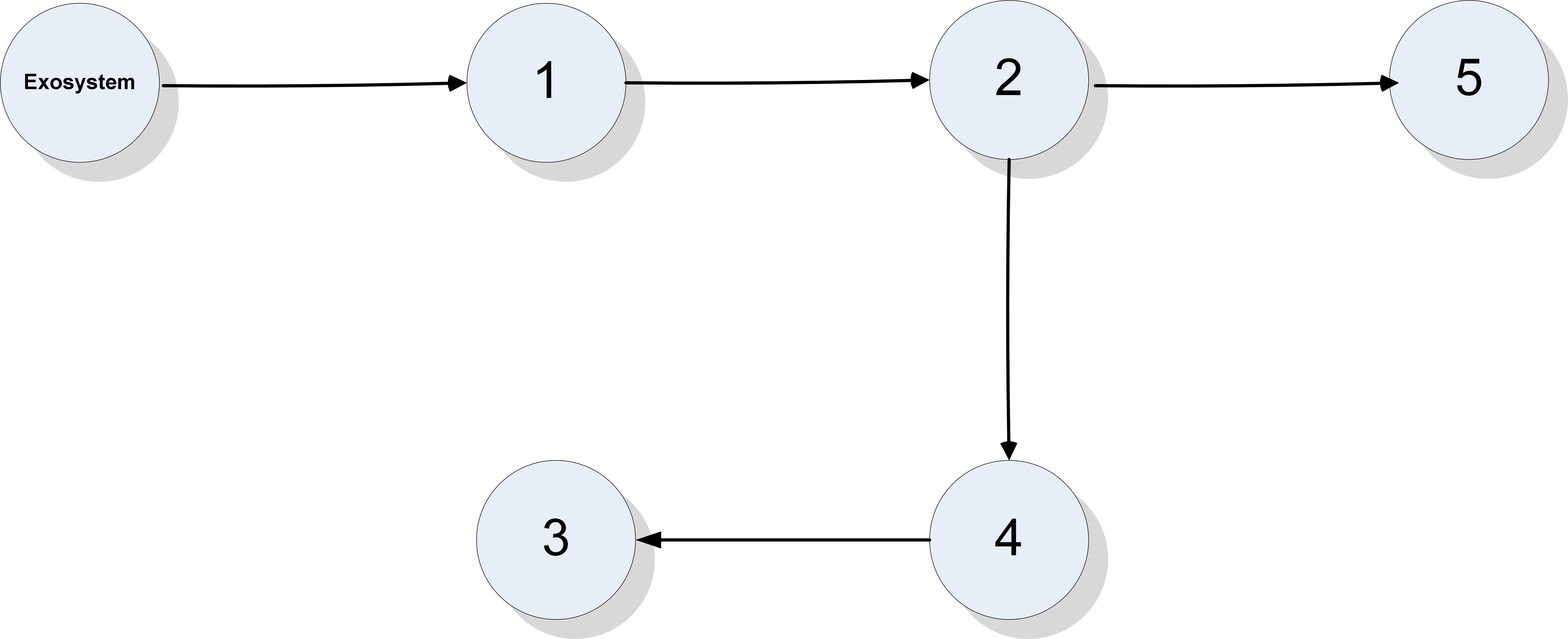}
	\centering
	\caption{Directed communication network with $5$ nodes}\label{graph_5nodes}
\end{figure}

\subsection{Numerical example}
\begin{figure}[t]
	\includegraphics[width=9cm, height=7cm]{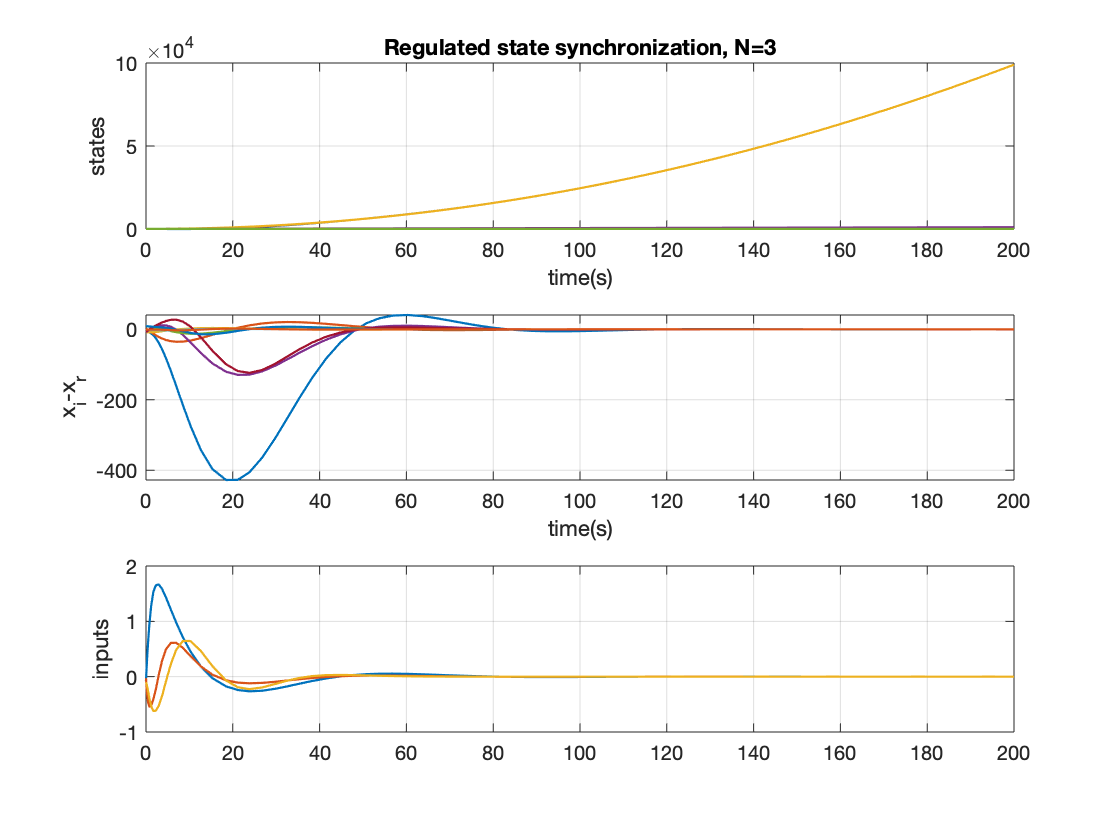}
	\centering
	\caption{Regulated state synchronization of MAS with full-state coupling with $N=3$}\label{Results_3nodes_Full}
\end{figure}

\begin{figure}[t]
	\includegraphics[width=9cm, height=7cm]{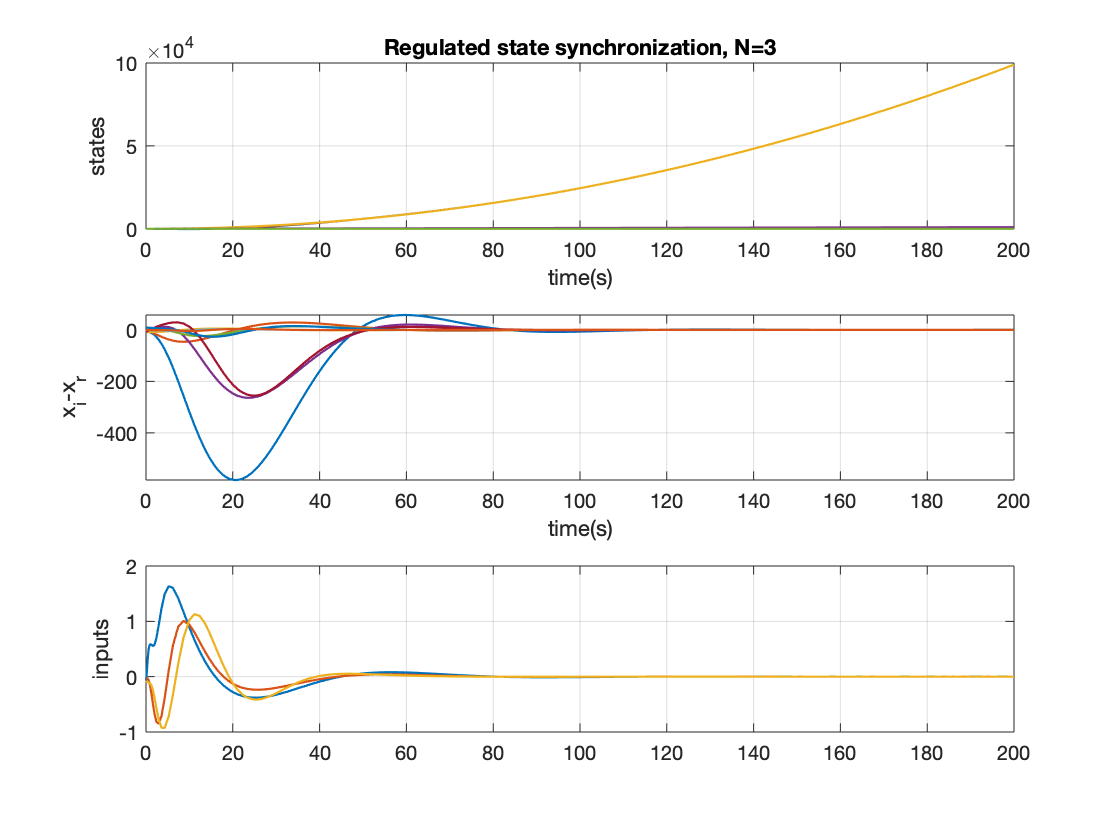}
	\centering
	\caption{Regulated state synchronization of MAS with partial-state coupling with $N=3$}\label{Results_3nodes}
\end{figure}

\begin{figure}[t]
	\includegraphics[width=9cm, height=7cm]{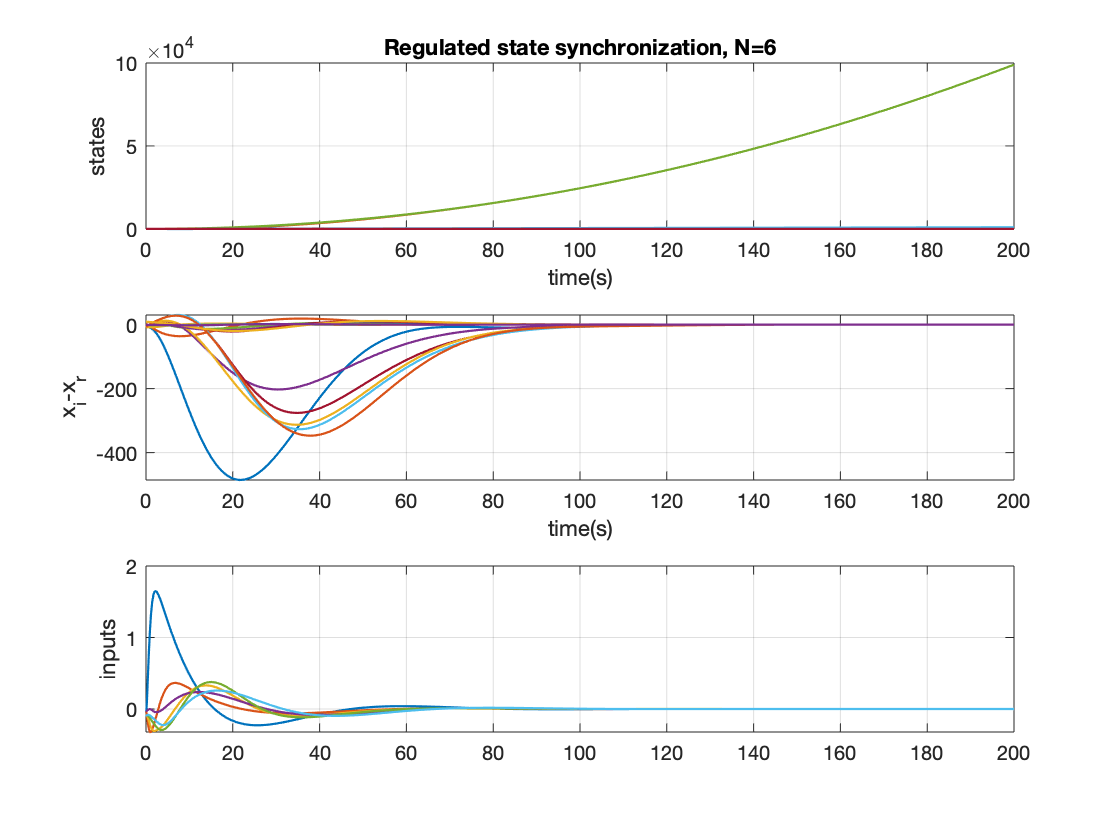}
	\centering
	\caption{Regulated state synchronization of MAS with full-state coupling with $N=6$}\label{Results_6nodes_Full}
\end{figure}

\begin{figure}[t]
	\includegraphics[width=9cm, height=7cm]{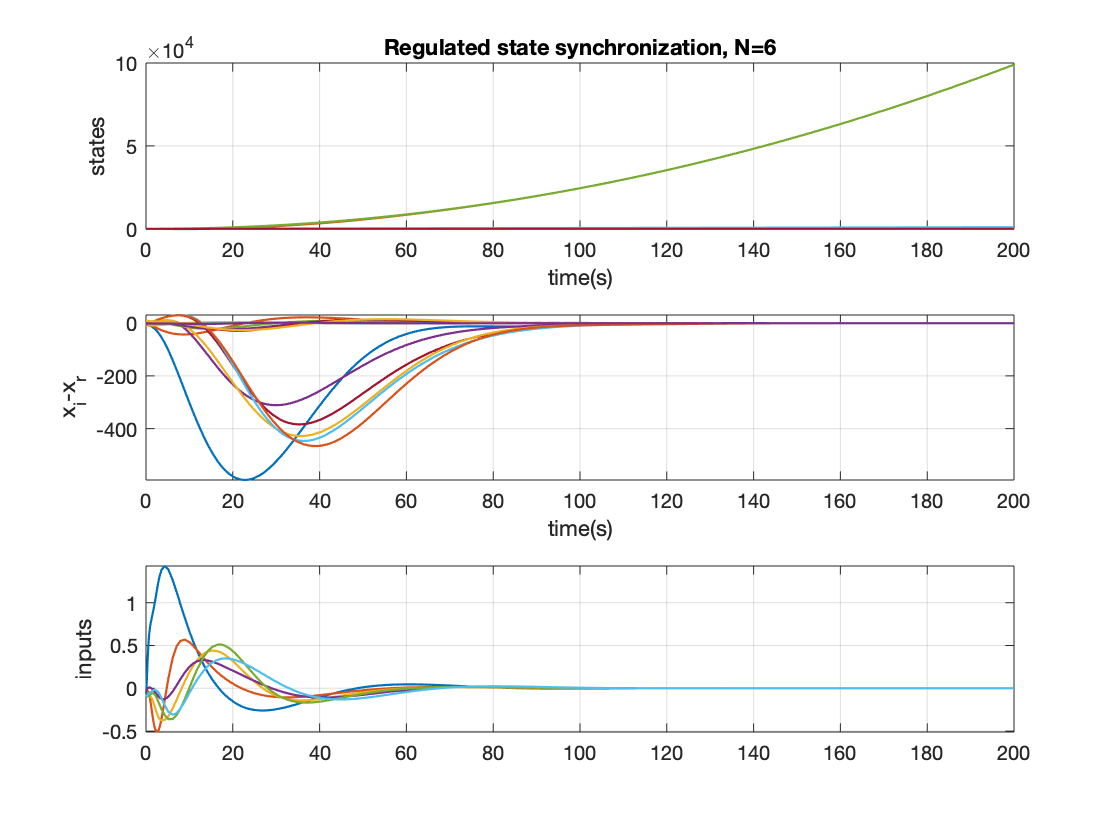}
	\centering
	\caption{Regulated state synchronization of MAS with partial-state coupling with $N=6$}\label{Results_6nodes}
\end{figure}

Consider the agents model \eqref{eq1} with full-state coupling as
\begin{equation*}\label{ex1}
\dot{x}_i(t)=\begin{pmatrix}
0&1&0\\0&0&1\\0&0&0
\end{pmatrix}x_i(t)+\begin{pmatrix}
0\\0\\1
\end{pmatrix}u_i(t-\tau_{i}),
\end{equation*}
and the agent models with partial-state coupling are as:
\begin{equation*}\label{ex2}
\begin{cases}
\dot{x}_i(t)=\begin{pmatrix}
0&1&0\\0&0&1\\0&0&0
\end{pmatrix}x_i(t)+\begin{pmatrix}
0\\0\\1
\end{pmatrix}u_i(t-\tau_{i}),\\
y_i(t)=\begin{pmatrix}
1&0&0
\end{pmatrix}x_i(t)
\end{cases}
\end{equation*}
In the following four cases, we simulate the regulated state synchronization via protocols \eqref{pscp1} and \eqref{pscp3} for MAS with full- and partial-state coupling respectively. Matrix $K$ in all the four cases with partial-state coupling is chosen as $K\T=\begin{pmatrix}6&11&6\end{pmatrix}$. Parameter $\rho$ is chosen as $\rho=1$ in all protocols of full- and partial-state coupling cases. Moreover, with the choise of $\eps=10^{-5}$ matrix $P_\eps$ is obtained by solving ARE \eqref{arespecial} as
\[
P_\eps=\begin{pmatrix}
0.0001&    0.0009 &   0.0032\\
0.0009   & 0.0096 &   0.0432\\
0.0032   & 0.0432  &  0.2941\\
\end{pmatrix}.
\]

\begin{figure}[t]
	\includegraphics[width=9cm, height=7cm]{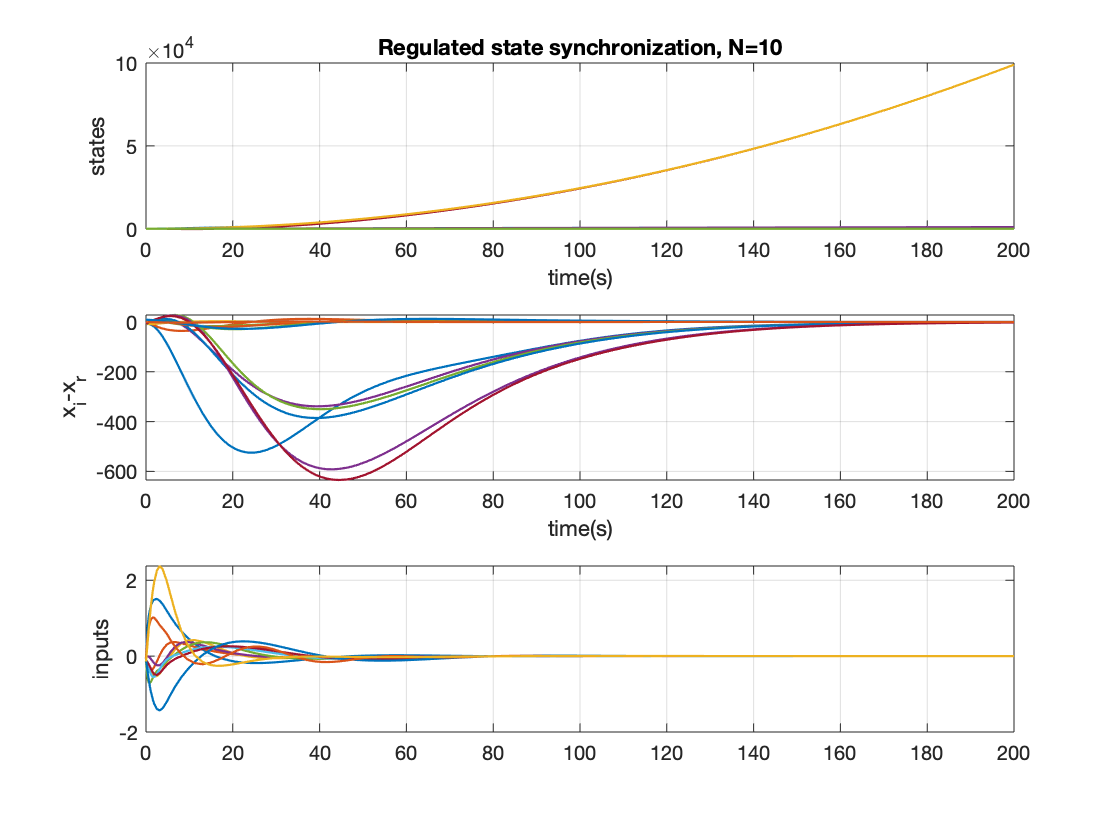}
	\centering
	\caption{Regulated state synchronization of MAS with full-state coupling with $N=10$}\label{Results_10nodes_Full}
\end{figure}

\begin{figure}[t]
	\includegraphics[width=9cm, height=7cm]{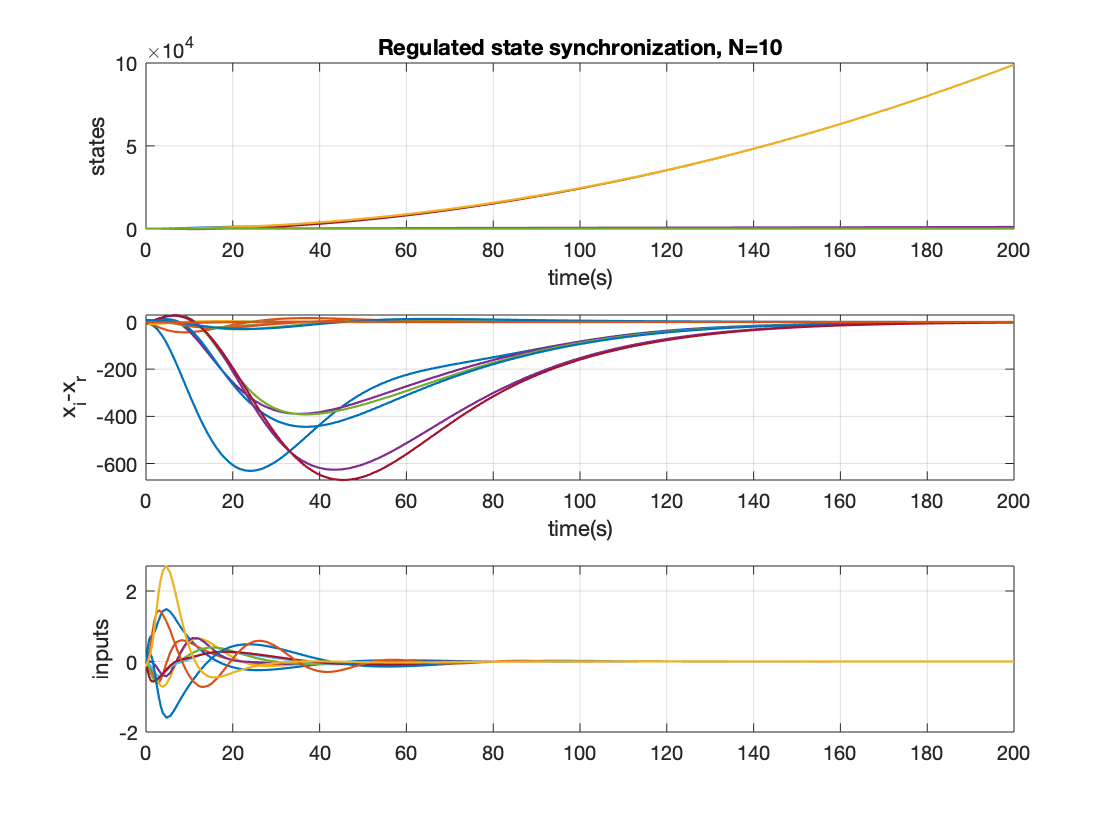}
	\centering
	\caption{Regulated state synchronization of MAS with partial-state coupling with $N=10$}\label{Results_10nodes}
\end{figure}

\begin{figure}[t]
	\includegraphics[width=9cm, height=7cm]{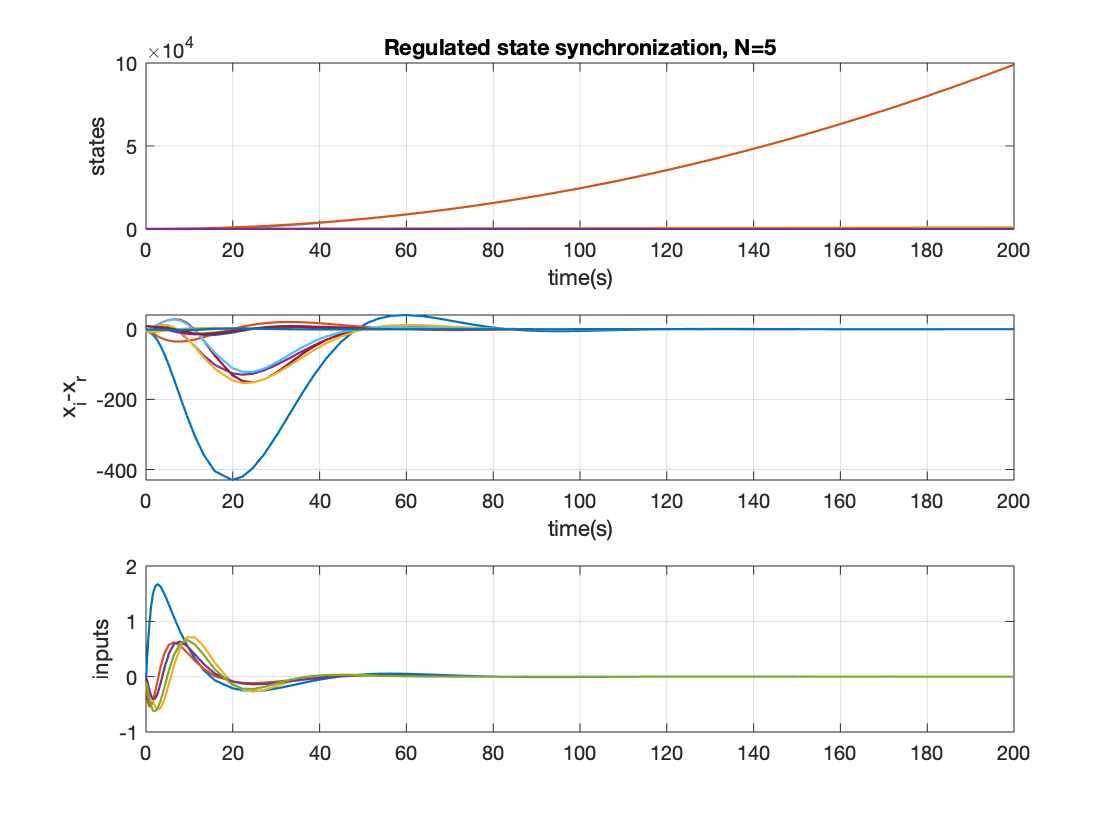}
	\centering
	\caption{Regulated state synchronization of MAS with full-state coupling with $N=5$}\label{Results_5nodes_Full}
\end{figure}

\begin{figure}[t]
	\includegraphics[width=9cm, height=7cm]{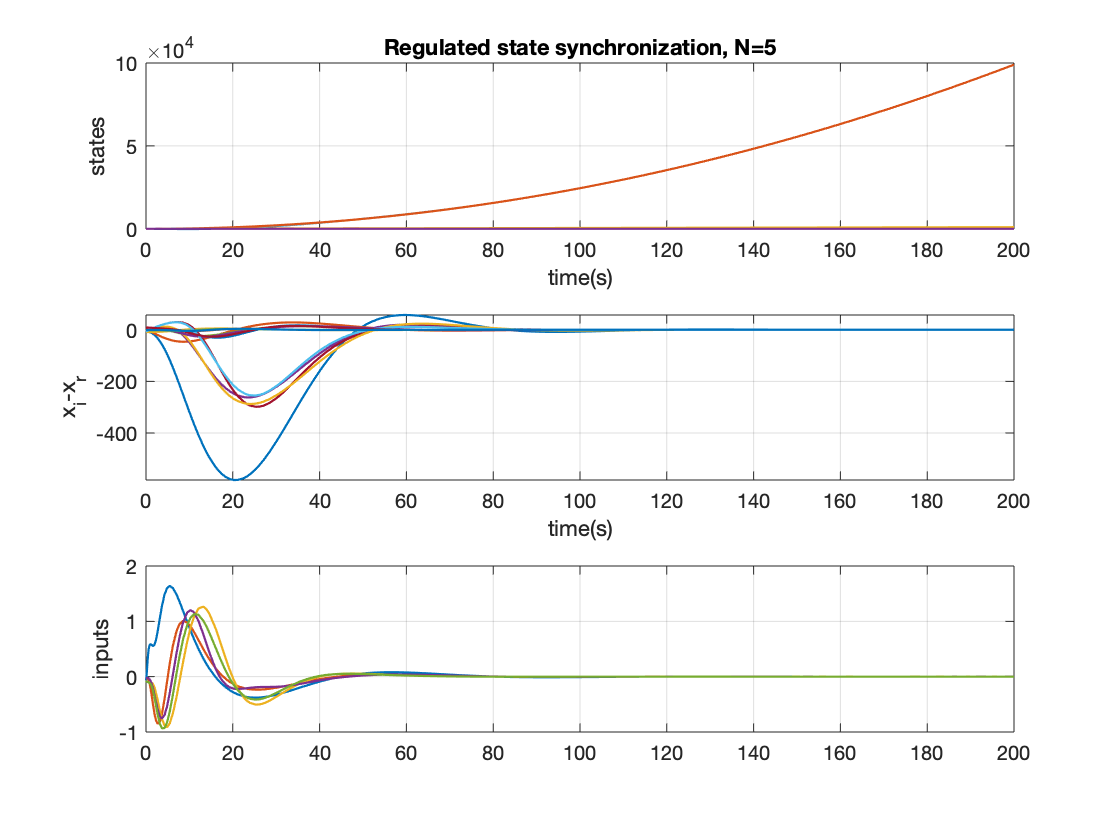}
	\centering
	\caption{Regulated state synchronization of MAS with partial-state coupling with $N=5$}\label{Results_5nodes}
\end{figure}

\begin{enumerate}
	\item  Firstly, we consider a MAS with $3$ agents, $N =3$ and communication graph shown in Figure \ref{graph_3nodes}. In this example, delays are as $\tau_1=1sec$, $\tau_2=2sec$, and $\tau_3=3sec$ for both networks with full- and partial-state coupling. The results of state synchronization of MAS with full-state coupling via protocol \eqref{pscp1} and partial-state coupling via protocol \eqref{pscp3} are presented in Figure \ref{Results_3nodes_Full} and \ref{Results_3nodes}, respectively.
	
	\item Next, we consider a MAS with $6$ agents, $N=6$, and communication topology with the associated graph shown in figure \ref{graph_6nodes}. The delays are chosen as $\tau_1=1sec$, $\tau_2=\tau_4=2sec$, $\tau_3=\tau_5=3sec$ and $\tau_6=1.5sec$ for both networks with full- and partial-state coupling. The simulation results for regulated state synchronization of MAS with full-state coupling and partial-state coupling are presented in Figure \ref{Results_6nodes_Full} and \ref{Results_6nodes}, respectively.
	
	\item In this case, we consider a MAS with $10$ agents, $N=10$, and communication topology with the associated graph shown in figure \ref{graph_10nodes}. The delays in both networks with full- and partial-state coupling are chosen as $\tau_1=1sec$, $\tau_2=\tau_4=2sec$, $\tau_3=\tau_5=3sec$, $\tau_6=1.5sec$, $\tau_7=0.5sec$, $\tau_8=0.7sec$, $\tau_9=4sec$, and $\tau_{10}=2.5sec$. The simulation results for these networks are illustrated in Figure \ref{Results_10nodes_Full} and \ref{Results_10nodes}.
	
	\item Finally, the simulation results for state synchronization of a MAS with $5$ agents and communication topology shown in Figure \ref{graph_5nodes}, are illustrated in Figure \ref{Results_5nodes_Full} and \ref{Results_5nodes} where the values of delays are as $\tau_1=1sec$, $\tau_2=\tau_4=2sec$, $\tau_3=\tau_5=3sec$.
\end{enumerate}

The simulation results show that our one-shot protocol designs do not need any knowledge of the communication network and achieve regulated state synchronization for any network with any number of agents. Moreover, we observe that upper bounds on the input delay tolerance only depends on the agent dynamics. 
\bibliographystyle{plain}
\bibliography{referenc}

\begin{thebibliography}{10}

\bibitem{cao-yu-ren-chen}
Y.~Cao, W.~Yu, W.~Ren, and G.~Chen.
\newblock An overview of recent progress in the study of distributed
  multi-agent coordination.
\newblock {\em IEEE Trans. on Industrial Informatics}, 9(1):427--438, 2013.

\bibitem{chowdhury-khalil}
D.~Chowdhury and H.~K. Khalil.
\newblock Synchronization in networks of identical linear systems with reduced
  information.
\newblock In {\em American Control Conference}, pages 5706--5711, Milwaukee,
  WI, 2018.

\bibitem{royle-godsil}
C.~Godsil and G.~Royle.
\newblock {\em Algebraic graph theory}, volume 207 of {\em Graduate Texts in
  Mathematics}.
\newblock Springer-Verlag, New York, 2001.

\bibitem{kim-shim-back-seo}
H.~Kim, H.~Shim, J.~Back, and J.~Seo.
\newblock Consensus of output-coupled linear multi-agent systems under fast
  switching network: averaging approach.
\newblock {\em Automatica}, 49(1):267--272, 2013.

\bibitem{liu-stoorvogel-saberi-zhang-cdc17}
Z.~Liu, A.~A. Stoorvogel, A.~Saberi, and M.~Zhang.
\newblock State synchronization of homogeneous continuous-time multi-agent
  systems with time-varying communication topology in presence of input delay.
\newblock In {\em American Control Conference}, pages 2273--2278, Seattle, WA,
  2017.

\bibitem{liu-zhang-saberi-stoorvogel-auto}
Z.~Liu, M.~Zhang, A.~Saberi, and A.~A. Stoorvogel.
\newblock State synchronization of multi-agent systems via static or adaptive
  nonlinear dynamic protocols.
\newblock {\em Automatica}, 95:316--327, 2018.

\bibitem{liu-zhang-saberi-stoorvogel-ejc}
Z.~Liu, M.~Zhang, A.~Saberi, and A.A. Stoorvogel.
\newblock Passivity based state synchronization of homogeneous discrete-time
  multi-agent systems via static protocol in the presence of input delay.
\newblock {\em European Journal of Control}, 41:16--24, 2018.

\bibitem{saber-murray3}
R.~Olfati-Saber, J.A. Fax, and R.M. Murray.
\newblock Consensus and cooperation in networked multi-agent systems.
\newblock {\em Proc. of the IEEE}, 95(1):215--233, 2007.

\bibitem{saber-murray2}
R.~Olfati-Saber and R.M. Murray.
\newblock Consensus problems in networks of agents with switching topology and
  time-delays.
\newblock {\em IEEE Trans. Aut. Contr.}, 49(9):1520--1533, 2004.

\bibitem{ren}
W.~Ren.
\newblock On consensus algorithms for double-integrator dynamics.
\newblock {\em IEEE Trans. Aut. Contr.}, 53(6):1503--1509, 2008.

\bibitem{ren-beard}
W.~Ren and R.W. Beard.
\newblock Consensus seeking in multiagent systems under dynamically changing
  interaction topologies.
\newblock {\em IEEE Trans. Aut. Contr.}, 50(5):655--661, 2005.

\bibitem{ren-book}
W.~Ren and Y.C. Cao.
\newblock {\em Distributed coordination of multi-agent networks}.
\newblock Communications and Control Engineering. Springer-Verlag, London,
  2011.

\bibitem{scardovi-sepulchre}
L.~Scardovi and R.~Sepulchre.
\newblock Synchronization in networks of identical linear systems.
\newblock {\em Automatica}, 45(11):2557--2562, 2009.

\bibitem{seo-back-kim-shim-iet}
J.H. Seo, J.~Back, H.~Kim, and H.~Shim.
\newblock Output feedback consensus for high-order linear systems having
  uniform ranks under switching topology.
\newblock {\em IET Control Theory and Applications}, 6(8):1118--1124, 2012.

\bibitem{seo-shim-back}
J.H. Seo, H.~Shim, and J.~Back.
\newblock Consensus of high-order linear systems using dynamic output feedback
  compensator: low gain approach.
\newblock {\em Automatica}, 45(11):2659--2664, 2009.

\bibitem{stoorvogel-saberi-acc}
A.A. Stoorvogel and A.~Saberi.
\newblock Consensus in the network with nonuniform constant input delay.
\newblock In {\em American Control Conference}, pages 4106--4111, Chicago, IL,
  2015.

\bibitem{su-huang-tac}
Y.~Su and J.~Huang.
\newblock Stability of a class of linear switching systems with applications to
  two consensus problem.
\newblock {\em IEEE Trans. Aut. Contr.}, 57(6):1420--1430, 2012.

\bibitem{tian-liu}
Y.-P. Tian and C.-L. Liu.
\newblock Consensus of multi-agent systems with diverse input and communication
  delays.
\newblock {\em IEEE Trans. Aut. Contr.}, 53(9):2122--2128, 2008.

\bibitem{tian}
Y.-P. Tian and C.-L. Liu.
\newblock Robust consensus of multi-agent systems with diverse input delays and
  asymmetric interconnection perturbations.
\newblock {\em Automatica}, 45(5):1347--1353, 2009.

\bibitem{tuna1}
S.E. Tuna.
\newblock {LQR}-based coupling gain for synchronization of linear systems.
\newblock Available: arXiv:0801.3390v1, 2008.

\bibitem{tuna3}
S.E. Tuna.
\newblock Conditions for synchronizability in arrays of coupled linear systems.
\newblock {\em IEEE Trans. Aut. Contr.}, 55(10):2416--2420, 2009.

\bibitem{nonidentical-delay-c}
X.~Wang, A.~Saberi, A.~A. Stoorvogel, H.F. Grip, and T.~Yang.
\newblock Synchronization for heterogenous networks of introspective
  right-invertible agents with uniform constant communication delay.
\newblock In {\em Proc. 52nd CDC}, pages 5198--5203, Firenze, Italy, 2012.

\bibitem{consensus-identical-delay-c}
X.~Wang, A.~Saberi, A.A. Stoorvogel, H.F. Grip, and T.~Yang.
\newblock Consensus in the network with uniform constant communication delay.
\newblock {\em Automatica}, 49(8):2461--2467, 2013.

\bibitem{wieland-kim-allgower}
P.~Wieland, J.S. Kim, and F.~Allg\"ower.
\newblock On topology and dynamics of consensus among linear high-order agents.
\newblock {\em International Journal of Systems Science}, 42(10):1831--1842,
  2011.

\bibitem{wu-book}
C.W. Wu.
\newblock {\em Synchronization in complex networks of nonlinear dynamical
  systems}.
\newblock World Scientific Publishing Company, Singapore, 2007.

\bibitem{xiao-wang}
F.~Xiao and L.~Wang.
\newblock Consensus protocols for discrete-time multi-agent systems with
  time-varying delays.
\newblock {\em Automatica}, 44(10):2577--2582, 2008.

\bibitem{zhang-saberi-stoorvogel-cdc15}
M.~Zhang, A.~Saberi, and A.A. Stoorvogel.
\newblock Synchronization in a network of identical discrete-time agents with
  unknown, nonuniform constant input delay.
\newblock In {\em Proc. 54th CDC}, pages 7054--7059, Osaka, Japan, 2015.

\end{thebibliography}

\end{document}